\documentclass{llncs}

\usepackage[T1]{fontenc}
\usepackage[utf8]{inputenc}

\usepackage{amssymb,amsmath,graphicx,graphics,enumerate,tkz-graph}
\usepackage{microtype}

\usepackage{fullpage}

\title{On Colouring $(2P_2,H)$-Free and $(P_5,H)$-Free Graphs\thanks{This paper was supported by EPSRC (EP/K025090/1) and the Leverhulme Trust (RPG-2016-258).}}

\author{Konrad K. Dabrowski \and Dani\"el~Paulusma}
\institute{Department of Computer Science, Durham University,\\
Science Laboratories, South Road, Durham DH1 3LE, United Kingdom\\
\texttt{\{konrad.dabrowski,daniel.paulusma\}@durham.ac.uk}}

\newcommand{\NP}{{\sf NP}}

\newcounter{ctrclaim}[theorem]
\newcounter{ctrobs}[theorem]
\newcounter{ctrcase}[theorem]
\newcounter{ctrsubcase}[ctrcase]
\newcounter{ctrfake}

\newtheorem{oproblem}{Open Problem}

\newcommand\figurenames{Figs.}

\begin{document}
\maketitle
\begin{abstract}
The {\sc Colouring} problem asks whether the vertices of a graph can be coloured with at most~$k$ colours for a given integer~$k$ in such a way that no two adjacent vertices receive the same colour.
A graph is $(H_1,H_2)$-free if it has no induced subgraph isomorphic to~$H_1$ or~$H_2$.
A connected graph~$H_1$ is almost classified if {\sc Colouring} on $(H_1,H_2)$-free graphs is known to be polynomial-time solvable or \NP-complete for all but finitely many connected graphs~$H_2$.
We show that every connected graph~$H_1$ apart from the claw~$K_{1,3}$ and the $5$-vertex path~$P_5$ is almost classified.
We also prove a number of new hardness results for {\sc Colouring} on $(2P_2,H)$-free graphs.
This enables us to list all graphs~$H$ for which the complexity of {\sc Colouring} is open on $(2P_2,H)$-free graphs and all graphs~$H$ for which the complexity of {\sc Colouring} is open on $(P_5,H)$-free graphs.
In fact we show that these two lists coincide.
Moreover, we show that the complexities of {\sc Colouring} for $(2P_2,H)$-free graphs and for $(P_5,H)$-free graphs are the same for all known cases.
\end{abstract}

\section{Introduction}\label{s-intro}

Graph colouring is an extensively studied concept in both Computer Science and Mathematics due to its many application areas.
A {\em $k$-colouring} of a graph~$G=(V,E)$ is a mapping $c:V\to \{1,\dots,k\}$ such that $c(u)\neq c(v)$ whenever $uv\in E$.
The {\sc Colouring} problem that of deciding whether a given graph~$G$ has a $k$-colouring for a given integer~$k$.
If~$k$ is fixed, then we write $k$-{\sc Colouring} instead.
It is well known that even {\sc $3$-Colouring} is \NP-complete~\cite{Lo73}.

Due to the computational hardness of {\sc Colouring}, it is natural to restrict the input to special graph classes.
A class is hereditary if it closed under vertex deletion.
Hereditary graph classes form a large collection of well-known graph classes for which the {\sc Colouring} problem has been extensively studied.
A classical result in the area is due to Gr\"otschel, Lov\'asz, and Schrijver~\cite{GLS84}, who showed that {\sc Colouring} is polynomial-time solvable for perfect graphs.

Graphs with no induced subgraph isomorphic to a graph in a set~${\cal H}$ are said to be {\em ${\cal H}$-free}.
It is readily seen that a graph class~${\cal G}$ is hereditary if and only if it there exists a set~${\cal H}$ such that every graph in~${\cal G}$ is ${\cal H}$-free.
If the graphs of~${\cal H}$ are required to be minimal under taking induced subgraphs, then~${\cal H}$ is unique.
For example, the set~${\cal H}$ of minimal forbidden induced subgraphs for the class of perfect graphs consists of all odd holes and odd antiholes~\cite{CRST06}.

Kr\'al', Kratochv\'{\i}l, Tuza, and Woeginger~\cite{KKTW01} classified the complexity of {\sc Colouring} for the case where~${\cal H}$ consists of a single graph~$H$.
They proved that {\sc Colouring} on $H$-free graphs is polynomial-time solvable if~$H$ is an induced subgraph of~$P_4$ or $P_1+\nobreak P_3$ and \NP-complete otherwise.\footnote{We refer to Section~\ref{s-pre} for notation used throughout Section~\ref{s-intro}.}

Kr\'al' et al.~\cite{KKTW01} also initiated a complexity study of {\sc Colouring} for graph classes defined by two forbidden induced subgraphs~$H_1$ and~$H_2$.
Such graph classes are said to be {\em bigenic}.
For bigenic graph classes, no dichotomy is known or even conjectured, despite many results~\cite{BDJP16,BBKRS05,BGPS12a,CH,DDP15,DGP14,GHP,HH,HL,Hu16,KMP,KKTW01,LM15,Ma13,Ma,ML17,Sc05}.
For instance, if we forbid two graphs~$H_1$ and~$H_2$ with $|V(H_1)|\leq 4$ and $|V(H_2)|\leq 4$, then there are three open cases left, namely when 
$(H_1,H_2)\in \{(K_{1,3},4P_1),\allowbreak (K_{1,3},2P_1+\nobreak P_2),\allowbreak (C_4,4P_1)\}$ (see~\cite{LM15} and \figurename~\ref{fig:fourvertices}).
If~$H_1$ and~$H_2$ are connected with $|V(H_1)|\leq 5$ and $|V(H_2)|\leq 5$, then
there are four open cases left, namely when $H_1=P_5$ and $H_2\in \{\overline{C_3+2P_1},\overline{C_3+P_2},\overline{P_1+2P_2}\}$ (see~\cite{KMP} and \figurename~\ref{fig:fivevertices}) and when $H_1=K_{1,3}$ and $H_2=\overline{C_4+P_1}$ (see~\cite{ML17} and \figurename~\ref{fig:fivevertices}).
To give another example, Blanch\'e et al.~\cite{BDJP16} determined the complexity of {\sc Colouring} for $(H,\overline{H})$-free graphs for every graph~$H$ except when $H=P_3+\nobreak sP_1$ for $s\geq 3$ or $H=P_4+\nobreak sP_1$ for $s\geq 2$.

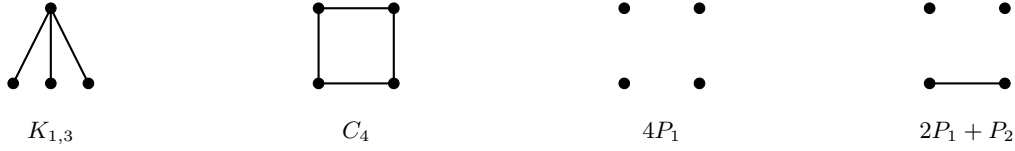
\begin{figure}
\begin{center}
\begin{tabular}{cccc}
\begin{minipage}{0.24\textwidth}
\begin{center}
\begin{tikzpicture}
\coordinate (x00) at (0,0) ;
\coordinate (x10) at (0.5,0) ;
\coordinate (x20) at (1,0) ;
\coordinate (x11) at (0.5,1) ;
\draw [fill=black] (x00) circle (2pt) ;
\draw [fill=black] (x10) circle (2pt) ;
\draw [fill=black] (x20) circle (2pt) ;
\draw [fill=black] (x11) circle (2pt) ;
\draw [thick] (x00) -- (x11);
\draw [thick] (x10) -- (x11);
\draw [thick] (x20) -- (x11);
\end{tikzpicture}
\end{center}
\end{minipage}
&
\begin{minipage}{0.24\textwidth}
\begin{center}
\begin{tikzpicture}
\coordinate (x00) at (0,0) ;
\coordinate (x01) at (0,1) ;
\coordinate (x10) at (1,0) ;
\coordinate (x11) at (1,1) ;
\draw [fill=black] (x00) circle (2pt) ;
\draw [fill=black] (x01) circle (2pt) ;
\draw [fill=black] (x10) circle (2pt) ;
\draw [fill=black] (x11) circle (2pt) ;
\draw [thick] (x00) -- (x01) -- (x11) -- (x10) -- (x00);
\end{tikzpicture}
\end{center}
\end{minipage}
&
\begin{minipage}{0.24\textwidth}
\begin{center}
\begin{tikzpicture}
\coordinate (x00) at (0,0) ;
\coordinate (x01) at (0,1) ;
\coordinate (x10) at (1,0) ;
\coordinate (x11) at (1,1) ;
\draw [fill=black] (x00) circle (2pt) ;
\draw [fill=black] (x01) circle (2pt) ;
\draw [fill=black] (x10) circle (2pt) ;
\draw [fill=black] (x11) circle (2pt) ;
\end{tikzpicture}
\end{center}
\end{minipage}
&
\begin{minipage}{0.24\textwidth}
\begin{center}
\begin{tikzpicture}
\coordinate (x00) at (0,0) ;
\coordinate (x01) at (0,1) ;
\coordinate (x10) at (1,0) ;
\coordinate (x11) at (1,1) ;
\draw [fill=black] (x00) circle (2pt) ;
\draw [fill=black] (x01) circle (2pt) ;
\draw [fill=black] (x10) circle (2pt) ;
\draw [fill=black] (x11) circle (2pt) ;
\draw [thick] (x00) -- (x10);
\end{tikzpicture}
\end{center}
\end{minipage}
\\
\\
$K_{1,3}$ & $C_4$ & $4P_1$ & $2P_1+\nobreak P_2$
\end{tabular}
\end{center}
\caption{\label{fig:fourvertices}The graphs from the three pairs $(H_1,H_2)\in \{(K_{1,3},4P_1),\allowbreak (K_{1,3},2P_1+\nobreak P_2),\allowbreak (C_4,4P_1)\}$ of graphs on at most four vertices, for which the complexity of {\sc Colouring} on $(H_1,H_2)$-free graphs is still open.}
\end{figure}

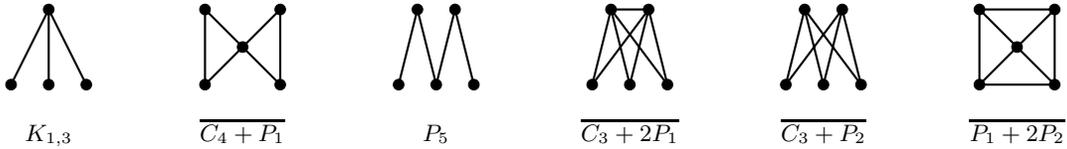
\begin{figure}
\begin{center}
\begin{tabular}{cccccc}
\begin{minipage}{0.15\textwidth}
\begin{center}
\begin{tikzpicture}
\coordinate (x00) at (0,0) ;
\coordinate (x10) at (0.5,0) ;
\coordinate (x20) at (1,0) ;
\coordinate (x11) at (0.5,1) ;
\draw [fill=black] (x00) circle (2pt) ;
\draw [fill=black] (x10) circle (2pt) ;
\draw [fill=black] (x20) circle (2pt) ;
\draw [fill=black] (x11) circle (2pt) ;
\draw [thick] (x00) -- (x11);
\draw [thick] (x10) -- (x11);
\draw [thick] (x20) -- (x11);
\end{tikzpicture}
\end{center}
\end{minipage}
&
\begin{minipage}{0.15\textwidth}
\begin{center}
\begin{tikzpicture}
\coordinate (x00) at (0,0) ;
\coordinate (x01) at (0,1) ;
\coordinate (x10) at (1,0) ;
\coordinate (x11) at (1,1) ;
\coordinate (x22) at (0.5,0.5) ;
\draw [fill=black] (x00) circle (2pt) ;
\draw [fill=black] (x01) circle (2pt) ;
\draw [fill=black] (x10) circle (2pt) ;
\draw [fill=black] (x11) circle (2pt) ;
\draw [fill=black] (x22) circle (2pt) ;
\draw [thick] (x00) -- (x01);
\draw [thick] (x10) -- (x11);
\draw [thick] (x00) -- (x22) -- (x11);
\draw [thick] (x01) -- (x22) -- (x10);
\end{tikzpicture}
\end{center}
\end{minipage}
&
\begin{minipage}{0.15\textwidth}
\begin{center}
\begin{tikzpicture}
\coordinate (x0) at (0,0) ;
\coordinate (y0) at (0.25,1) ;
\coordinate (x1) at (0.5,0) ;
\coordinate (y1) at (0.75,1) ;
\coordinate (x2) at (1,0) ;
\draw [fill=black] (x0) circle (2pt) ;
\draw [fill=black] (y0) circle (2pt) ;
\draw [fill=black] (x1) circle (2pt) ;
\draw [fill=black] (y1) circle (2pt) ;
\draw [fill=black] (x2) circle (2pt) ;
\draw [thick] (x0) -- (y0) -- (x1) -- (y1) -- (x2);
\end{tikzpicture}
\end{center}
\end{minipage}
&
\begin{minipage}{0.15\textwidth}
\begin{center}
\begin{tikzpicture}
\coordinate (x0) at (0,0) ;
\coordinate (y0) at (0.25,1) ;
\coordinate (x1) at (0.5,0) ;
\coordinate (y1) at (0.75,1) ;
\coordinate (x2) at (1,0) ;
\draw [fill=black] (x0) circle (2pt) ;
\draw [fill=black] (y0) circle (2pt) ;
\draw [fill=black] (x1) circle (2pt) ;
\draw [fill=black] (y1) circle (2pt) ;
\draw [fill=black] (x2) circle (2pt) ;
\draw [thick] (y1) -- (x0) -- (y0) -- (x1) -- (y1) -- (x2) -- (y0) -- (y1);
\end{tikzpicture}
\end{center}
\end{minipage}
&
\begin{minipage}{0.15\textwidth}
\begin{center}
\begin{tikzpicture}
\coordinate (x0) at (0,0) ;
\coordinate (y0) at (0.25,1) ;
\coordinate (x1) at (0.5,0) ;
\coordinate (y1) at (0.75,1) ;
\coordinate (x2) at (1,0) ;
\draw [fill=black] (x0) circle (2pt) ;
\draw [fill=black] (y0) circle (2pt) ;
\draw [fill=black] (x1) circle (2pt) ;
\draw [fill=black] (y1) circle (2pt) ;
\draw [fill=black] (x2) circle (2pt) ;
\draw [thick] (y1) -- (x0) -- (y0) -- (x1) -- (y1) -- (x2) -- (y0);
\end{tikzpicture}
\end{center}
\end{minipage}
&
\begin{minipage}{0.15\textwidth}
\begin{center}
\begin{tikzpicture}
\coordinate (x00) at (0,0) ;
\coordinate (x01) at (0,1) ;
\coordinate (x10) at (1,0) ;
\coordinate (x11) at (1,1) ;
\coordinate (x22) at (0.5,0.5) ;
\draw [fill=black] (x00) circle (2pt) ;
\draw [fill=black] (x01) circle (2pt) ;
\draw [fill=black] (x10) circle (2pt) ;
\draw [fill=black] (x11) circle (2pt) ;
\draw [fill=black] (x22) circle (2pt) ;
\draw [thick] (x00) -- (x01) -- (x11) -- (x10) -- (x00);
\draw [thick] (x00) -- (x22) -- (x11);
\draw [thick] (x01) -- (x22) -- (x10);
\end{tikzpicture}
\end{center}
\end{minipage}
\\
\\
$K_{1,3}$ & $\overline{C_4+P_1}$ & $P_5$ & $\overline{C_3+2P_1}$ & $\overline{C_3+P_2}$ & $\overline{P_1+2P_2}$
\end{tabular}
\end{center}
\caption{\label{fig:fivevertices}
The graphs from the four pairs $(H_1,H_2)\in \{(K_{1,3},\overline{C_4+P_1}),(P_5,\overline{C_3+2P_1}),(P_5,\overline{C_3+P_2}),(P_5,\overline{P_1+2P_2})\}$ of connected graphs on at most five vertices, for which the complexity of {\sc Colouring} on $(H_1,H_2)$-free graphs is still open.}
\end{figure}

The related problems {\sc Precolouring Extension} and {\sc List Colouring} have also been studied for bigenic graph classes.
For the first problem, we are given a graph~$G$, an integer~$k$ and a $k$-colouring~$c'$ defined on an induced subgraph of~$G$.
The question is whether~$G$ has a $k$-colouring~$c$ extending~$c'$.
For the second problem, each vertex~$u$ of the input graph~$G$ has a list~$L(u)$ of colours.
Here the question is whether~$G$ has a colouring~$c$ that respects~$L$, that is, with $c(u)\in L(u)$ for all $u\in V(G)$.
For the {\sc Precolouring Extension} problem no classification is known and we refer to the survey~\cite{GJPS} for an overview on what is known.
In contrast to the incomplete classifications for {\sc Colouring} and {\sc Precolouring Extension}, Golovach and Paulusma~\cite{GP14} showed a dichotomy for
the complexity of {\sc List Colouring} on bigenic graph classes.

\medskip
\noindent
{\bf Our Approach.}
To get a handle on the computational complexity classification of {\sc Colouring} for bigenic graph classes, we continue the line of research in~\cite{BBKRS05,HL,KMP,Ma13,Ma,ML17} by considering pairs $(H_1,H_2)$, where~$H_1$ and~$H_2$ are both connected.
We introduce the following notion.
We say that a connected graph~$H_1$ is {\em almost classified} if {\sc Colouring} on $(H_1,H_2)$-free graphs is known to be either polynomial-time solvable or \NP-complete for all but finitely many connected graphs~$H_2$.
This leads to the following research question:
\begin{center}
{\em Which connected graphs are almost classified?}
\end{center}

\medskip
\noindent
{\bf Our Results.}
In Section~\ref{s-almost} we show, by combining known results from the literature, that every connected graph~$H_1$ apart from the claw~$K_{1,3}$ and the $5$-vertex path~$P_5$ is almost classified.
In fact we show that the number of pairs $(H_1,H_2)$ of connected graphs for which the complexity of {\sc Colouring} is unknown is finite if neither~$H_1$ nor~$H_2$ is isomorphic to~$K_{1,3}$ or~$P_5$.
In Section~\ref{s-hard} we prove a number of new hardness results for {\sc Colouring} restricted to $(2P_2,H_2)$-free graphs (which form a subclass of $(P_5,H_2)$-free graphs).
We do the latter by adapting the \NP-hardness construction from~\cite{GH09} for {\sc List Colouring} restricted to complete bipartite graphs.
In Section~\ref{s-con}, we first summarize our knowledge on the complexity of {\sc Colouring} restricted to $(2P_2,H)$-free graphs and $(P_5,H)$-free graphs.
Afterwards, we list all graphs~$H$ for which the complexity of {\sc Colouring} on $(2P_2,H)$-free graphs is still open, and all graphs~$H$ for which the complexity of {\sc Colouring} on $(P_5,H)$-free graphs is still open.
As it turns out, these two lists coincide.
Moreover, the complexities of {\sc Colouring} for $(2P_2,H)$-free graphs and for $(P_5,H)$-free graphs turn out to be the same for all cases that are known.

\section{Preliminaries}\label{s-pre}
We consider only finite, undirected graphs without multiple edges or self-loops.
Let $G=(V,E)$ be a graph.
The {\em complement}~$\overline{G}$ of~$G$ is the graph with vertex set~$V(G)$ and an edge between two distinct vertices if and only if these two vertices are not adjacent in~$G$.
For a subset $S\subseteq V$, we let~$G[S]$ denote the subgraph of~$G$ {\em induced} by~$S$, which has vertex set~$S$ and edge set $\{uv\; |\; u,v\in S, uv\in E\}$.

Let $\{H_1,\ldots,H_p\}$ be a set of graphs.
A graph~$G$ is {\em $(H_1,\ldots,H_p)$-free} if~$G$ has no induced subgraph isomorphic to a graph in $\{H_1,\ldots,H_p\}$.
If~$p=1$, we may write $H_1$-free instead of $(H_1)$-free.
The {\em disjoint union} $G+\nobreak H$ of two vertex-disjoint graphs~$G$ and~$H$ is the graph $(V(G)\cup V(H), E(G)\cup E(H))$.
The disjoint union of~$r$ copies of a graph~$G$ is denoted by~$rG$.
A {\em linear forest} is the disjoint union of one or more paths.

The graphs~$C_r$, $K_r$ and~$P_r$ denote the cycle, complete graph and path on~$r$ vertices, respectively.
The graph~$K_3$ is also known as the {\em triangle}.
The graph~$K_{r,s}$ denotes the complete bipartite graph with partition classes of size~$r$ and~$s$, respectively.
The graph~$K_{1,3}$ is also called the {\em claw}.

The graph~$S_{h,i,j}$, for $1\leq h\leq i\leq j$, denotes the {\em subdivided claw}, that is, the tree that has only one vertex~$x$ of degree~$3$ and exactly three leaves, which are of distance~$h$,~$i$ and~$j$ from~$x$, respectively.
Observe that $S_{1,1,1}=K_{1,3}$. The graph~$S_{1,1,2}$ is also known as the {\em fork} or the {\em chair}.

The graph~$T_{h,i,j}$ with $0\leq h\leq i\leq j$ denotes the graph with vertices $a_0,\ldots,a_h$, $b_0,\ldots,b_i$ and $c_0,\ldots,c_j$ and edges $a_0b_0$, $b_0c_0$, $c_0a_0$, $a_pa_{p+1}$ for $p \in \{0,\ldots, h-\nobreak1\}$, $b_pb_{p+1}$ for $p\in\{0,\ldots,i-\nobreak 1\}$ and $c_pc_{p+1}$ for $p\in\{0,\ldots,j-\nobreak 1\}$.
Note that $T_{0,0,0}=C_3$.
The graph $T_{0,0,1}=\overline{P_1+P_3}$ is known as the {\em paw}, the graph~$T_{0,1,1}$ as the {\em bull}, the graph~$T_{1,1,1}$ as the {\em net}, and the graph~$T_{0,0,2}$ is known as the {\em hammer}; see also \figurename~\ref{fig:Thij-examples}.
Also note that~$T_{h,i,j}$ is the line graph of~$S_{h+1,i+1,j+1}$.

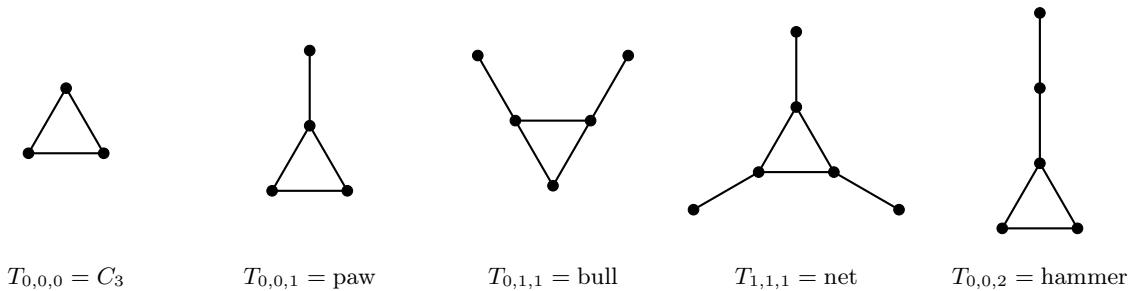
\begin{figure}
\begin{center}
\begin{tabular}{ccccc}
\begin{minipage}{0.19\textwidth}
\begin{center}
\begin{tikzpicture}
\coordinate (a0) at (0.5,0) ;
\coordinate (b0) at (-0.5,0) ;
\coordinate (c0) at (0,0.86602540378) ;
\draw [fill=black] (a0) circle (2pt) ;
\draw [fill=black] (b0) circle (2pt) ;
\draw [fill=black] (c0) circle (2pt) ;
\draw [thick] (a0) -- (c0) -- (b0) -- (a0);
\end{tikzpicture}
\end{center}
\end{minipage}
&
\begin{minipage}{0.19\textwidth}
\begin{center}
\begin{tikzpicture}
\coordinate (a0) at (0.5,0) ;
\coordinate (b0) at (-0.5,0) ;
\coordinate (c0) at (0,0.86602540378) ;
\coordinate (c1) at (0,1.86602540378) ;
\draw [fill=black] (a0) circle (2pt) ;
\draw [fill=black] (b0) circle (2pt) ;
\draw [fill=black] (c0) circle (2pt) ;
\draw [fill=black] (c1) circle (2pt) ;
\draw [thick] (c1) -- (c0) -- (b0) -- (a0) -- (c0);
\end{tikzpicture}
\end{center}
\end{minipage}
&
\begin{minipage}{0.19\textwidth}
\begin{center}
\begin{tikzpicture}
\coordinate (a0) at (0,0) ;
\coordinate (b0) at (60:1) ;
\coordinate (c0) at (120:1) ;
\coordinate (b1) at (60:2) ;
\coordinate (c1) at (120:2) ;
\draw [fill=black] (a0) circle (2pt) ;
\draw [fill=black] (b0) circle (2pt) ;
\draw [fill=black] (c0) circle (2pt) ;
\draw [fill=black] (b1) circle (2pt) ;
\draw [fill=black] (c1) circle (2pt) ;
\draw [thick] (c1) -- (c0) -- (b0) -- (a0) -- (c0);
\draw [thick] (b1) -- (b0);
\end{tikzpicture}
\end{center}
\end{minipage}
&
\begin{minipage}{0.19\textwidth}
\begin{center}
\begin{tikzpicture}
\coordinate (a0) at (90:0.57735026919) ;
\coordinate (b0) at (90+120:0.57735026919) ;
\coordinate (c0) at (90+240:0.57735026919) ;
\coordinate (a1) at (90:1.57735026919) ;
\coordinate (b1) at (90+120:1.57735026919) ;
\coordinate (c1) at (90+240:1.57735026919) ;
\draw [fill=black] (a0) circle (2pt) ;
\draw [fill=black] (b0) circle (2pt) ;
\draw [fill=black] (c0) circle (2pt) ;
\draw [fill=black] (a1) circle (2pt) ;
\draw [fill=black] (b1) circle (2pt) ;
\draw [fill=black] (c1) circle (2pt) ;
\draw [thick] (c1) -- (c0) -- (b0) -- (a0) -- (c0);
\draw [thick] (a1) -- (a0);
\draw [thick] (b1) -- (b0);
\end{tikzpicture}
\end{center}
\end{minipage}
&
\begin{minipage}{0.19\textwidth}
\begin{center}
\begin{tikzpicture}
\coordinate (a0) at (0.5,0) ;
\coordinate (b0) at (-0.5,0) ;
\coordinate (c0) at (0,0.86602540378) ;
\coordinate (c1) at (0,1.86602540378) ;
\coordinate (c2) at (0,2.86602540378) ;
\draw [fill=black] (a0) circle (2pt) ;
\draw [fill=black] (b0) circle (2pt) ;
\draw [fill=black] (c0) circle (2pt) ;
\draw [fill=black] (c1) circle (2pt) ;
\draw [fill=black] (c2) circle (2pt) ;
\draw [thick] (c2) -- (c1) -- (c0) -- (b0) -- (a0) -- (c0);
\end{tikzpicture}
\end{center}
\end{minipage}
\\
\\
$T_{0,0,0}=C_3$ & $T_{0,0,1}=\mbox{paw}$ &  $T_{0,1,1}=\mbox{bull}$ &  $T_{1,1,1}=\mbox{net}$ & $T_{0,0,2}=\mbox{hammer}$
\end{tabular}
\end{center}
\caption{\label{fig:Thij-examples}Examples of~$T_{h,i,j}$ graphs.}
\end{figure}

Let~${\cal T}$ be the class of graphs every component of which is isomorphic to a graph~$T_{h,i,j}$ for some $1\leq h\leq i\leq j$ or a path~$P_r$ for some $r\geq 1$.
The following result, which is due to Schindl and which we use in Section~\ref{s-con}, shows that the~$T_{h,i,j}$ graphs play an important role for our study.

\begin{theorem}[\cite{Sc05}]\label{t-schindl}
For $p\geq 1$, let $H_1,\ldots,H_p$ be graphs whose complement is not in~${\cal T}$.
Then {\sc Colouring} is \NP-complete for $(H_1,\ldots,H_p)$-free graphs.
\end{theorem}

\section{Almost Classified Graphs}\label{s-almost}

In this section we prove the following result, from which it immediately follows that every connected graph apart from~$K_{1,3}$ and~$P_5$ is almost classified.
In Section~\ref{s-con} we discuss why~$K_{1,3}$ and~$P_5$ are not almost classified.

\begin{theorem}\label{t-almost1}
There are only finitely many pairs $(H_1,H_2)$ of connected graphs with $\{H_1,H_2\} \cap \{K_{1,3},P_5\}=\emptyset$, such that the complexity of {\sc Colouring} on $(H_1,H_2)$-free graphs is unknown.
\end{theorem}

\begin{proof}
We first make a useful observation.
Let~$H$ be a tree that is not isomorphic to~$K_{1,3}$ or~$P_5$ and that is not an induced subgraph of~$P_4$.
If~$H$ contains a vertex of degree at least~$4$ then it contains an induced~$K_{1,4}$.
If~$H$ has maximum degree~$3$, then since~$H$ is connected and not isomorphic to~$K_{1,3}$, it must contain an induced~$S_{1,1,2}$.
If~$H$ has maximum degree at most~$2$, then it is a path, and since it is not isomorphic to~$P_5$ and not an induced subgraph of~$P_4$, it follows that~$H$ must be a path on at least six vertices.
We conclude that if~$H$ is a tree that is not isomorphic to~$K_{1,3}$ or~$P_5$ and that is not an induced subgraph of~$P_4$, then~$H$ contains~$K_{1,4}$ or~$S_{1,1,2}$ as an induced subgraph or~$H$ is a path on at least six vertices.

\medskip
\noindent
Now let $(H_1,H_2)$ be a pair of connected graphs with $\{H_1,H_2\} \cap \{K_{1,3},P_5\}=\emptyset$.
If~$H_1$ or~$H_2$ is an induced subgraph of~$P_4$, then {\sc Colouring} is polynomial-time solvable for $(H_1,H_2)$-free graphs, as {\sc Colouring} is polynomial-time solvable for $P_4$-free graphs (see, for example,~\cite{KKTW01}).
Hence we may assume that this is not the case.
If~$H_1$ and~$H_2$ both contain at least one cycle~\cite{EHK98} or both contain an induced~$K_{1,3}$~\cite{Ho81}, then even {\sc $3$-Colouring} is \NP-complete.
Hence we may also assume that at least one of $H_1,H_2$ is a tree and that at least one of $H_1,H_2$ is a $K_{1,3}$-free graph.
This leads, without loss of generality, to the following two cases.

\medskip
\noindent
{\bf Case 1.} $H_1$ is a tree and $K_{1,3}$-free.\\
Then~$H_1$ is a path.
First suppose that~$H_1$ has at least~22 vertices.
It is known that {\sc $4$-Colouring} is \NP-complete for $(P_{22},C_3)$-free graphs~\cite{HJP14} and that {\sc Colouring} is \NP-complete for $(P_9,C_4)$-free graphs~\cite{GHP} and for $(2P_2,C_r)$-free graphs for all $r \geq 5$~\cite{KKTW01}.
Hence we may assume that~$H_2$ is a tree.
By the observation at the start of the proof, this implies that~$H_2$ contains an induced~$K_{1,4}$, $S_{1,1,2}$ or~$P_6$.
Therefore~$H$ contains an induced~$4P_1$ or~$2P_1+\nobreak P_2$.
Since~$H_1$ is a path on at least~22 vertices, $H_1$ contains an induced~$4P_1$.
As {\sc Colouring} is \NP-complete for $(4P_1,2P_1+\nobreak P_2)$-free graphs~\cite{KKTW01}, {\sc Colouring} is \NP-complete for $(H_1,H_2)$-free graphs.

Now suppose that~$H_1$ has at most~21 vertices.
By the observation at the start of the proof, $H_1$ is a path on at least six vertices.
It is known that {\sc $5$-Colouring} is \NP-complete for $P_6$-free graphs~\cite{Hu16}.
As~$K_6$ is not $5$-colourable, this means that {\sc $5$-Colouring} is \NP-complete for $(P_6,K_6)$-free graphs, as observed in~\cite{GJPS}.
Therefore we may assume that~$H_2$ is $K_6$-free.
Recall that {\sc Colouring} is \NP-complete for $(2P_1+\nobreak P_2,4P_1)$-free graphs~\cite{KKTW01}, which are contained in the class of $(P_6,4P_1)$-free graphs.
Therefore we may assume that~$H_2$ is $4P_1$-free.
Since~$H_2$ is $(K_6,4P_1)$-free, Ramsey's Theorem~\cite{Ra30} implies that~$|V(H_2)|$ is bounded by a constant.
We conclude that both~$H_1$ and~$H_2$ have size bounded by a constant.

\medskip
\noindent
{\bf Case 2.} $H_1$ is a tree and not $K_{1,3}$-free, and~$H_2$ is $K_{1,3}$-free and not a tree.\\
Then~$H_1$ contains a vertex of degree at least~$3$ and~$H_2$ contains an induced cycle~$C_r$ for some $r\geq 3$.
It is known that {\sc $3$-Colouring} is \NP-complete for $(K_{1,5},C_3)$-free graphs~\cite{MF96} and for $(K_{1,3},C_r)$-free graphs whenever $r \geq 4$~\cite{KKTW01}.
We may therefore assume that~$H_1$ is a tree of maximum degree at most~$4$ and that~$H_2$ contains at least one induced~$C_3$ but no induced cycles on more than three vertices.
Recall that {\sc $4$-Colouring} is \NP-complete for $(P_{22},C_3)$-free graphs~\cite{HJP14}.
Hence we may assume that~$H_1$ is a $P_{22}$-free tree.
As~$H_1$ has maximum degree at most~$4$, we find that~$H_1$ has a bounded number of vertices.

By assumption, $H_1$ contains a vertex of degree at least~$3$.
As {\sc Colouring} is \NP-complete for $(K_{1,3},K_4)$-free graphs~\cite{KKTW01}, we may assume that~$H_2$ is $K_4$-free.
By the observation at the start of the proof, $H_1$ must contain an induced~$K_{1,4}$ or~$S_{1,1,2}$.
Recall that {\sc Colouring} is \NP-complete for the class of $(2P_1+\nobreak P_2,4P_1)$-free graphs~\cite{KKTW01}, which is contained in the class of~$(K_{1,4},S_{1,1,2},4P_1)$-free graphs.
Hence we may assume that~$H_2$ is $4P_1$-free.
Since~$H_2$ is $(K_4,4P_1)$-free, Ramsey's Theorem~\cite{Ra30} implies that~$|V(H_2)|$ is bounded by a constant.
Again, we conclude that in this case both~$H_1$ and~$H_2$ have size bounded by a constant.\qed
\end{proof}

\begin{corollary}\label{t-almost}
Every connected graph apart from~$K_{1,3}$ and~$P_5$ is almost classified.
\end{corollary}

\section{Hardness Results}\label{s-hard}

In this section we prove that {\sc Colouring} restricted to $(2P_2,H)$-free graphs is \NP-complete for several graphs~$H$.
To prove our results we adapt a hardness construction from Golovach and Heggernes~\cite{GH09} for proving that {\sc List Colouring} is \NP-complete for complete bipartite graphs.
As observed in~\cite{GP14}, a minor modification of this construction yields that {\sc List Colouring} is \NP-complete for complete split graphs, which are the graphs obtained from complete bipartite graphs by changing one of the bipartition classes into a clique.

We first describe the construction of~\cite{GH09}, which uses a reduction from the \NP-complete~\cite{Sc78} problem {\sc Not-All-Equal $3$-Satisfiability} with positive literals only.
To define this problem, let $X= \{x_1,x_2,\ldots,x_n\}$ be a set of logical variables, and let ${\cal C} = \{C_1, C_2,\ldots, C_m\}$ be a set of $3$-literal clauses over~$X$ in which all literals are positive and every literal appears at most once in each clause.
The question is whether~$X$ has a truth assignment such that each clause in~${\cal C}$ contains at least one true literal and at least one false literal.
If so, we say that such a truth assignment is {\em satisfying}.

Let $(X,{\cal C})$ be an instance of {\sc Not-All-Equal $3$-Satisfiability} with positive literals only.
We construct an instance $(G_1,L)$ of {\sc List Colouring} as follows.
For each~$x_i$ we introduce a vertex, which we also denote by~$x_i$ and which we say is of {\em $x$-type}.
We define $L(x_1)=\{1,2\}$, $L(x_2)=\{3,4\},\ldots,L(x_n)=\{2n-1,2n\}$.
In this way, each~$x_i$ has one odd colour and one even colour in its list, and all lists~$L(x_i)$ are pairwise disjoint.
For each~$C_j$ we introduce two vertices, which we denote by~$C_j$ and~$C_j'$ and which we say are of {\em $C$-type}.
If $C_j=\{x_g,x_h,x_i\}$ with $L(x_g)=\{a,a+\nobreak 1\}$, $L(x_h)=\{b,b+\nobreak 1\}$ and $L(x_i)=\{c,c+\nobreak 1\}$, then we set $L(C_j)=\{a,b,c\}$ and $L(C_j')=\{a+\nobreak 1,b+\nobreak 1,c+\nobreak 1\}$.
Hence each~$C_j$ has only odd colours in its list and each~$C_j'$ has only even colours in its list.
To obtain the graph~$G_1$ we add an edge between every vertex of $x$-type and every vertex of $C$-type.
Note that~$G_1$ is a complete bipartite graph with bipartition classes $\{x_1,\ldots,x_n\}$ and $\{C_1,\ldots,C_m\}\cup \{C_1',\ldots,C_m'\}$.

We also construct an instance $(G_2,L)$ where~$G_2$ is obtained from~$G_1$ by adding edges between every pair of vertices of $x$-type.
Note that~$G_2$ is a complete split graph.

The following lemma is straightforward.
We refer to~\cite{GH09} for a proof for the case involving~$G_1$.
The case involving~$G_2$ follows from this proof and the fact that the lists~$L(x_i)$ are pairwise disjoint, as observed in~\cite{GP14}.

\begin{lemma}[\cite{GH09}]\label{l-true}
$({\cal C},X)$ has a satisfying truth assignment if and only if~$G_1$ has a colouring that respects~$L$ if and only if~$G_2$ has a colouring that respects~$L$.
\end{lemma}

We now extend~$G_1$ and~$G_2$ into graphs~$G'_1$ and~$G'_2$, respectively, by adding a clique~$K$ consisting of~$2n$ new vertices $k_1,\ldots,k_{2n}$ and by adding an edge between a vertex~$k_\ell$ and a vertex~$u$ of the original graph if and only if $\ell\notin L(u)$.
We say that the vertices $k_1,\ldots,k_{2n}$ are of {\em $k$-type}.

\begin{lemma}\label{l-truth}
$({\cal C},X)$ has a satisfying truth assignment if and only if~$G'_1$ has a $2n$-colouring if and only if~$G'_2$ has a $2n$-colouring.
\end{lemma}

\begin{proof}
Let $i\in \{1,2\}$.
By Lemma~\ref{l-true}, we only need to show that~$G_i$ has a colouring that respects~$L$ if and only if~$G'_i$ has a $2n$-colouring.
First suppose that~$G_i$ has a colouring~$c$ that respects~$L$.
We extend~$c$ to a colouring~$c'$ of~$G'_i$ by setting $c'(k_\ell)=\ell$ for $\ell \in\{1,\ldots,2n\}$.
Now suppose that~$G'_i$ has a $2n$-colouring~$c'$.
As the $k$-type vertices form a clique, we may assume without loss of generality that $c'(k_\ell)=\ell$ for $\ell \in \{1,\ldots,2n\}$.
Hence the restriction of~$c'$ to~$G_i$ yields a colouring~$c$ that respects~$L$.\qed
\end{proof}

In the next two lemmas we show forbidden induced subgraphs in~$G_1'$ and~$G_2'$, respectively.
The complements of these forbidden graphs are shown in \figurenames~\ref{fig:co-G1-forb} and~\ref{fig:co-G2-forb}, respectively.

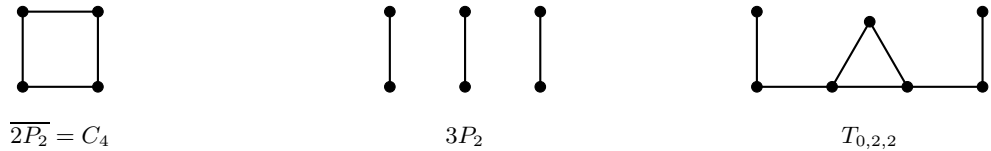
\begin{figure}
\begin{center}
\begin{tabular}{ccc}
\begin{minipage}{0.32\textwidth}
\begin{center}
\begin{tikzpicture}
\coordinate (x00) at (0,0) ;
\coordinate (x01) at (0,1) ;
\coordinate (x10) at (1,0) ;
\coordinate (x11) at (1,1) ;
\draw [fill=black] (x00) circle (2pt) ;
\draw [fill=black] (x01) circle (2pt) ;
\draw [fill=black] (x10) circle (2pt) ;
\draw [fill=black] (x11) circle (2pt) ;
\draw [thick] (x00) -- (x01) -- (x11) -- (x10) -- (x00);
\end{tikzpicture}
\end{center}
\end{minipage}
&
\begin{minipage}{0.32\textwidth}
\begin{center}
\begin{tikzpicture}
\coordinate (x00) at (0,0) ;
\coordinate (x01) at (0,1) ;
\coordinate (x10) at (1,0) ;
\coordinate (x11) at (1,1) ;
\coordinate (x20) at (2,0) ;
\coordinate (x21) at (2,1) ;
\draw [fill=black] (x00) circle (2pt) ;
\draw [fill=black] (x01) circle (2pt) ;
\draw [fill=black] (x10) circle (2pt) ;
\draw [fill=black] (x11) circle (2pt) ;
\draw [fill=black] (x20) circle (2pt) ;
\draw [fill=black] (x21) circle (2pt) ;
\draw [thick] (x00) -- (x01);
\draw [thick] (x10) -- (x11);
\draw [thick] (x20) -- (x21);
\end{tikzpicture}
\end{center}
\end{minipage}
&
\begin{minipage}{0.32\textwidth}
\begin{center}
\begin{tikzpicture}
\coordinate (a0) at (0.5,0) ;
\coordinate (a1) at (1.5,0) ;
\coordinate (a2) at (1.5,1) ;
\coordinate (b0) at (-0.5,0) ;
\coordinate (b1) at (-1.5,0) ;
\coordinate (b2) at (-1.5,1) ;
\coordinate (c0) at (0,0.86602540378) ;
\draw [fill=black] (a0) circle (2pt) ;
\draw [fill=black] (a1) circle (2pt) ;
\draw [fill=black] (a2) circle (2pt) ;
\draw [fill=black] (b0) circle (2pt) ;
\draw [fill=black] (b1) circle (2pt) ;
\draw [fill=black] (b2) circle (2pt) ;
\draw [fill=black] (c0) circle (2pt) ;
\draw [thick] (a2) -- (a1) -- (a0) -- (b0) -- (b1) -- (b2);
\draw [thick] (a0) -- (c0) -- (b0);
\end{tikzpicture}
\end{center}
\end{minipage}
\\
\\
$\overline{2P_2}=C_4$ & $3P_2$ & $T_{0,2,2}$
\end{tabular}
\end{center}
\caption{\label{fig:co-G1-forb}Graphs that are not induced subgraphs of the complement of~$G_1'$.}
\end{figure}

\begin{lemma}\label{l-s1}
The graph~$G_1'$ is $(2P_2,\overline{3P_2},\overline{T_{0,2,2}})$-free.
\end{lemma}

\begin{proof}
We will prove that~$\overline{G_1'}$ is $(C_4,3P_2,T_{0,2,2})$-free.
Observe that in~$\overline{G_1'}$, the set of $x$-type vertices is a clique, the set of $C$-type vertices is a clique and the set of $k$-type vertices is an independent set.
Furthermore, in~$\overline{G_1'}$, no $x$-type vertex is adjacent to a $C$-type vertex.

\medskip
\noindent
{\bf $\mathbf{C_4}$-freeness.}
For contradiction, suppose that~$\overline{G_1'}$ contains an induced subgraph~$H$ isomorphic to~$C_4$; say the vertices of~$H$ are $u_1,u_2,u_3,u_4$ in that order.
As the union of the set of $x$-type and $C$-type vertices induces a $P_3$-free graph in~$\overline{G_1'}$, there must be at least two vertices of the~$C_4$ that are neither $x$-type nor~$C$-type.
Since the $k$-type vertices form an independent set, we may assume without loss of generality that~$u_1$ and~$u_3$ are of $k$-type.
It follows that~$u_2$ and~$u_4$ cannot be of $k$-type.
As the set of $x$-type vertices and the set of $C$-type vertices each from a clique in~$\overline{G_1'}$, but~$u_2$ is non-adjacent to~$u_4$, we may assume without loss of generality that~$u_2$ is of $x$-type and~$u_4$ is of $C$-type.
Then~$u_4$ is adjacent to the two $k$-type neighbours of an $x$-type vertex, which correspond to an even and odd colour.
This is not possible as~$u_4$, being a $C$-type vertex, is adjacent in~$\overline{G_1'}$ to (exactly three) $k$-type vertices, which correspond either to even colours only or to odd colours only.
We conclude that~$\overline{G_1'}$ is $C_4$-free.

\medskip
\noindent
{\bf $\mathbf{3P_2}$-freeness.}
For contradiction, suppose that~$\overline{G_1'}$ contains an induced subgraph~$H$ isomorphic to~$3P_2$.
As the $C$-type vertices and $x$-type vertices each form a clique in~$\overline{G_1'}$, one edge of~$H$ must consist of two $k$-type vertices.
This is not possible, as $k$-type vertices form an independent set in~$\overline{G_1'}$.
We conclude that~$\overline{G_1'}$ is $3P_2$-free.

\medskip
\noindent
{\bf $\mathbf{T_{0,2,2}}$-freeness.}
For contradiction, suppose that~$\overline{G_1'}$ contains an induced subgraph~$H$ isomorphic to~$T_{0,2,2}$ with vertices $a_0,a_1,a_2,b_0,b_1,b_2,c_0$ and edges~$a_0b_0$, $b_0c_0$, $c_0a_0$, $a_0a_1$, $a_1a_2$, $b_0b_1$, $b_1b_2$.
As the $k$-type vertices form an independent set in~$\overline{G_1'}$, at least one of $a_1,a_2$ and at least one of $b_1,b_2$ is of $x$-type or $C$-type.
As the $x$-type vertices and the $C$-type vertices form cliques in~$\overline{G_1'}$, we may assume without loss of generality that at least one of $a_1,a_2$ is of $C$-type and at least one of $b_1,b_2$ is of $x$-type.
As the $C$-type vertices and the $x$-type vertices each form a clique in~$\overline{G_1'}$, this means that~$c_0$ must be of $k$-type, $a_0$ cannot be of $x$-type and~$b_0$ cannot be of $C$-type.
As $k$-type vertices form an independent set in~$\overline{G_1'}$, $a_0$ and~$b_0$ cannot be of $k$-type.
Therefore~$a_0$ is of~$C$-type and~$b_0$ is of $x$-type.
This is a contradiction, as $C$-type vertices are non-adjacent to $x$-type vertices.
We conclude that~$\overline{G_1'}$ is $T_{0,2,2}$-free.\qed
\end{proof}

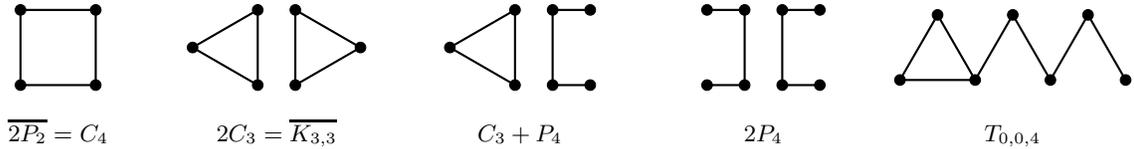
\begin{figure}
\begin{center}
\begin{tabular}{ccccc}
\begin{minipage}{0.15\textwidth}
\begin{center}
\begin{tikzpicture}
\coordinate (x00) at (0,0) ;
\coordinate (x01) at (0,1) ;
\coordinate (x10) at (1,0) ;
\coordinate (x11) at (1,1) ;
\draw [fill=black] (x00) circle (2pt) ;
\draw [fill=black] (x01) circle (2pt) ;
\draw [fill=black] (x10) circle (2pt) ;
\draw [fill=black] (x11) circle (2pt) ;
\draw [thick] (x00) -- (x01) -- (x11) -- (x10) -- (x00);
\end{tikzpicture}
\end{center}
\end{minipage}
&
\begin{minipage}{0.19\textwidth}
\begin{center}
\begin{tikzpicture}
\coordinate (x00) at (0,0) ;
\coordinate (x01) at (0,1) ;
\coordinate (x10) at (0.5,0) ;
\coordinate (x11) at (0.5,1) ;
\coordinate (x02) at (-0.86602540378,0.5) ;
\coordinate (x12) at (1.36602540378,0.5) ;
\draw [fill=black] (x00) circle (2pt) ;
\draw [fill=black] (x01) circle (2pt) ;
\draw [fill=black] (x10) circle (2pt) ;
\draw [fill=black] (x11) circle (2pt) ;
\draw [fill=black] (x02) circle (2pt) ;
\draw [fill=black] (x12) circle (2pt) ;
\draw [thick] (x00) -- (x01) -- (x02) -- (x00);
\draw [thick] (x10) -- (x11) -- (x12) -- (x10);
\end{tikzpicture}
\end{center}
\end{minipage}
&
\begin{minipage}{0.19\textwidth}
\begin{center}
\begin{tikzpicture}
\coordinate (x00) at (0,0) ;
\coordinate (x01) at (0,1) ;
\coordinate (x02) at (-0.86602540378,0.5) ;
\draw [fill=black] (x00) circle (2pt) ;
\draw [fill=black] (x01) circle (2pt) ;
\draw [fill=black] (x02) circle (2pt) ;
\draw [thick] (x00) -- (x01) -- (x02) -- (x00);
\coordinate (x20) at (0.5,0) ;
\coordinate (x21) at (0.5,1) ;
\coordinate (x30) at (1,0) ;
\coordinate (x31) at (1,1) ;
\draw [fill=black] (x20) circle (2pt) ;
\draw [fill=black] (x21) circle (2pt) ;
\draw [fill=black] (x30) circle (2pt) ;
\draw [fill=black] (x31) circle (2pt) ;
\draw [thick] (x31) -- (x21) -- (x20) -- (x30);
\end{tikzpicture}
\end{center}
\end{minipage}
&
\begin{minipage}{0.19\textwidth}
\begin{center}
\begin{tikzpicture}
\coordinate (x00) at (0,0) ;
\coordinate (x01) at (0,1) ;
\coordinate (x10) at (0.5,0) ;
\coordinate (x11) at (0.5,1) ;
\coordinate (x20) at (1,0) ;
\coordinate (x21) at (1,1) ;
\coordinate (x30) at (1.5,0) ;
\coordinate (x31) at (1.5,1) ;
\draw [fill=black] (x00) circle (2pt) ;
\draw [fill=black] (x01) circle (2pt) ;
\draw [fill=black] (x10) circle (2pt) ;
\draw [fill=black] (x11) circle (2pt) ;
\draw [fill=black] (x20) circle (2pt) ;
\draw [fill=black] (x21) circle (2pt) ;
\draw [fill=black] (x30) circle (2pt) ;
\draw [fill=black] (x31) circle (2pt) ;
\draw [thick] (x01) -- (x11) -- (x10) -- (x00);
\draw [thick] (x31) -- (x21) -- (x20) -- (x30);
\end{tikzpicture}
\end{center}
\end{minipage}
&
\begin{minipage}{0.20\textwidth}
\begin{center}
\begin{tikzpicture}
\coordinate (a0) at (0.5,0) ;
\coordinate (a1) at (1,0.86602540378) ;
\coordinate (a2) at (1.5,0) ;
\coordinate (a3) at (2,0.86602540378) ;
\coordinate (a4) at (2.5,0) ;
\coordinate (b0) at (-0.5,0) ;
\coordinate (c0) at (0,0.86602540378) ;
\draw [fill=black] (a0) circle (2pt) ;
\draw [fill=black] (a1) circle (2pt) ;
\draw [fill=black] (a2) circle (2pt) ;
\draw [fill=black] (a3) circle (2pt) ;
\draw [fill=black] (a4) circle (2pt) ;
\draw [fill=black] (b0) circle (2pt) ;
\draw [fill=black] (c0) circle (2pt) ;
\draw [thick] (a4) -- (a3) -- (a2) -- (a1) -- (a0) -- (b0);
\draw [thick] (a0) -- (c0) -- (b0);
\end{tikzpicture}
\end{center}
\end{minipage}
\\
\\
$\overline{2P_2}=C_4$ & $2C_3=\overline{K_{3,3}}$ & $C_3+\nobreak P_4$ & $2P_4$ & $T_{0,0,4}$
\end{tabular}
\end{center}
\caption{\label{fig:co-G2-forb}Graphs that are not induced subgraphs of the complement of~$G_2'$.}
\end{figure}

\begin{lemma}\label{l-s2}
The graph~$G_2'$ is $(2P_2,\overline{2C_3},\overline{C_3+P_4},\overline{2P_4},\overline{T_{0,0,4}})$-free.
\end{lemma}

\begin{proof}
We will prove that~$\overline{G_2'}$ is $(C_4,2C_3,C_3+\nobreak P_4,2P_4,T_{0,0,4})$-free.
Observe that in~$\overline{G_2'}$, the set of $C$-type vertices is a clique, the set of $x$-type vertices is an independent set and the set of $k$-type vertices is an independent set.
Furthermore, in~$\overline{G_2'}$, no $x$-type vertex is adjacent to a $C$-type vertex and every $x$-type vertex has degree~$2$.
In fact, the union of the set of $x$-type vertices and the $k$-type vertices induces a disjoint union of $P_3$s in~$\overline{G_2'}$.

\medskip
\noindent
{\bf $\mathbf{C_4}$-freeness.}
For contradiction, suppose that~$\overline{G_2'}$ contains an induced subgraph~$H$ isomorphic to~$C_4$, say the vertices of~$H$ are $u_1,u_2,u_3,u_4$ in that order.
As the union of the set of $x$-type and $C$-type vertices induces a $P_3$-free graph in~$\overline{G_2'}$, there must be at least two vertices of the~$C_4$ that are neither $x$-type nor~$C$-type.
Since the $k$-type vertices form an independent set, we may assume without loss of generality that~$u_1$ and~$u_3$ are of $k$-type.
It follows that~$u_2$ and~$u_4$ cannot be of $k$-type.
As the set of vertices of~$C$-type form a clique in~$\overline{G_2'}$, at least one of $u_2,u_4$, say~$u_2$, is of $x$-type.
If~$u_4$ is also of $x$-type, then~$u_2$ and~$u_4$ are $x$-type vertices with the same two colours in their list, namely those corresponding to~$u_1$ and~$u_3$.
This is not possible.
Thus~$u_4$ must be of $C$-type.
Then~$u_4$ is adjacent to the two $k$-type neighbours of an $x$-type vertex, which correspond to an even and odd colour.
This is not possible as~$u_4$, being a $C$-type vertex, is adjacent in~$\overline{G_2'}$ to (exactly three) $k$-type vertices, which correspond either to even colours only or to odd colours only.
We conclude that~$\overline{G_2'}$ is $C_4$-free.

\medskip
\noindent
{\bf $\mathbf{(2C_3,C_3+\nobreak P_4,2P_4)}$-freeness.}
For contradiction, suppose that~$\overline{G_2'}$ contains an induced subgraph~$H$ isomorphic to $2C_3$, $C_3+\nobreak P_4$ or~$2P_4$.
As the $k$-type and $x$-type vertices induce a disjoint union of~$P_3$s in~$\overline{G_2'}$, both components of~$H$ must contain a $C$-type vertex.
This is not possible, as $C$-type vertices form a clique in~$\overline{G_2'}$.
We conclude that~$\overline{G_2'}$ is $(2C_3,C_3+\nobreak P_4,2P_4)$-free.

\medskip
\noindent
{\bf $\mathbf{T_{0,0,4}}$-freeness.}
For contradiction, suppose that~$\overline{G_2'}$ contains an induced subgraph~$H$ isomorphic to~$T_{0,0,4}$ with vertices $a_0,a_1,a_2,a_3,a_4,b_0,c_0$ and edges $a_0b_0$, $b_0c_0$, $c_0a_0$, $a_0a_1$, $a_1a_2$, $a_2a_3$, $a_3a_4$.

First suppose, that neither~$b_0$ nor~$c_0$ is of~$C$-type.
Since the union of the set of $x$-type vertices and the $k$-type vertices induces a disjoint union of $P_3$s in~$\overline{G_2'}$ it follows that~$a_0$ is of $C$-type.
Since~$b_0$ and~$c_0$ are not of $C$-type and no vertex of $C$-type has a neighbour of $x$-type in~$\overline{G_2'}$, it follows that~$b_0$ and~$c_0$ must be of $k$-type.
This is not possible, because the $k$-type vertices form an independent set in~$\overline{G_2'}$.

Now suppose that at least one of $b_0,c_0$ is of $C$-type.
Since the vertices of $C$-type induce a clique in~$\overline{G_2'}$, it follows that no vertex in $A:=\{a_1,a_2,a_3,a_4\}$ is of $C$-type.
Since~$A$ induces a~$P_4$ in~$\overline{G_2'}$, but the union of the set of $x$-type vertices and the $k$-type vertices induces a disjoint union of $P_3$s in~$\overline{G_2'}$, this is a contradiction.
We conclude that~$\overline{G_2'}$ is $T_{0,0,4}$-free.\qed
\end{proof}

We are now ready to state the two main results of this section.
It is readily seen that {\sc Colouring} belongs to \NP.
Then the first theorem follows from Lemma~\ref{l-truth} combined with Lemma~\ref{l-s1}, whereas the second one follows from Lemma~\ref{l-truth} combined with Lemma~\ref{l-s2}.
Note that~$\overline{2C_3}$ is isomorphic to~$K_{3,3}$.

\begin{theorem}\label{t-1}
{\sc Colouring} is \NP-complete for $(2P_2,\overline{3P_2},\overline{T_{0,2,2}})$-free graphs.
\end{theorem}

\begin{theorem}\label{t-2}
{\sc Colouring} is \NP-complete for $(2P_2,\overline{2C_3},\overline{C_3+P_4},\overline{2P_4},\overline{T_{0,0,4}})$-free graphs.
\end{theorem}

\section{Conclusions}\label{s-con}

We showed that every connected graph is almost classified except for the claw and the~$P_5$.
Our notion of almost classified graphs originated from recent work~\cite{HL,KMP,Ma,ML17} on {\sc Colouring} for $(H_1,H_2)$-free graphs for connected graphs~$H_1$ and~$H_2$, in particular when $H_1=P_5$.
We decreased the number of open cases for the latter graph by showing new \NP-hardness results for $(2P_2,H)$-free graphs.
In the following theorem we summarize all known results for {\sc Colouring} restricted to $(2P_2,H)$-free graphs and $(P_5,H)$-free graphs.

\begin{theorem}\label{t-p5}
Let~$H$ be a graph on~$n$ vertices.
Then the following two statements hold:
\begin{enumerate}[(i)]
\item If~$\overline{H}$ contains a graph in $\{C_3+\nobreak P_4,3P_2,2P_4\}$ as an induced subgraph, or~$\overline{H}$ is not an induced subgraph of $T_{1,1,3}+\nobreak P_{2n-1}$, then {\sc Colouring} is \NP-complete for $(2P_2,H)$-free graphs.\\[-8pt]
\item If~$\overline{H}$ is an induced subgraph of a graph in $\{2P_1+\nobreak P_3,P_1+\nobreak P_4,P_2+\nobreak P_3,P_5,T_{0,0,1}+\nobreak P_1,\allowbreak T_{0,1,1},T_{0,0,2}\}$ or of $sP_1+\nobreak P_2$ for some integer $s\geq 0$, then {\sc Colouring} is polynomial-time solvable for $(P_5,H)$-free graphs.
\end{enumerate}
\end{theorem}

\begin{proof}
If~$\overline{H}$ contains a graph in $\{2C_3,C_3+\nobreak P_4,3P_2,2P_4,T_{0,2,2},T_{0,0,4}\}$ as an induced subgraph, then {\sc Colouring} is \NP-complete for $(2P_2,H)$-free graphs due to Theorems~\ref{t-1} and~\ref{t-2}.
We may therefore assume that~$\overline{H}$ is $(2C_3,C_3+\nobreak P_4,3P_2,2P_4,T_{0,2,2},T_{0,0,4})$-free.
(Note that $T_{1,1,3}+\nobreak P_{2n-1}$ is $(2C_3,T_{0,2,2},T_{0,0,4})$-free.)

Recall that~${\cal T}$ is the class of graphs every component of which is isomorphic to a graph~$T_{h,i,j}$ for some $1\leq h\leq i\leq j$ or a path~$P_r$ for some $r\geq 1$.
Note that $\overline{2P_2}=C_4 \notin {\cal T}$.
Therefore, if $\overline{H} \notin {\cal T}$, then {\sc Colouring} is \NP-complete for $(2P_2,H)$-free graphs by Theorem~\ref{t-schindl}.
We may therefore assume that $\overline{H} \in {\cal T}$. 
Since~$\overline{H}$ is $2C_3$-free, $\overline{H}$ can contain at most one component that is not a path.
Since~$\overline{H}$ is $(T_{0,2,2},T_{0,0,4})$-free, if~$\overline{H}$ does have a component that is not a path, then this component must be an induced subgraph of~$T_{1,1,3}$.
The union of components of~$\overline{H}$ that are isomorphic to paths form an induced subgraph of~$P_{2n-1}$.
Therefore~$\overline{H}$ is an induced subgraph of $T_{1,1,3}+\nobreak P_{2n-1}$.

It is known that {\sc Colouring} is polynomial-time solvable for $(P_5,H)$-free graphs if~$\overline{H}$ is an induced subgraph of
$2P_1+\nobreak P_3$~\cite{Ma},
$P_1+\nobreak P_4$ (this follows from the fact that $(P_5,\overline{P_1+P_4})$-free graphs have clique-width at most~$5$~\cite{BLM04}; see also~\cite{BBKRS05} for a linear-time algorithm),
$P_2+\nobreak P_3$~\cite{ML17},
$P_5$~\cite{HL},
$T_{0,0,1}+\nobreak P_1$~\cite{KMP},
$T_{0,1,1}$~\cite{KMP},
$T_{0,0,2}$~\cite{KMP}
or $sP_1+\nobreak P_2$ for some integer $s\geq 0$~\cite{ML17}.\qed
\end{proof}

Theorem~\ref{t-p5} leads to the following open problem, which shows how the~$P_5$ is not almost classified.
Recall that $T_{0,0,0}=C_3$.

\begin{oproblem}\label{o-p5}
Determine the complexity of {\sc Colouring} for $(2P_2,H)$-free graphs and for $(P_5,H)$-free graphs if
\begin{itemize}
\item $\overline{H}=sP_1+\nobreak P_t+\nobreak T_{h,i,j}$ for $0\leq h\leq i\leq j \leq 1$, $s\geq 0$ and $2\leq t\leq 3$
\item $\overline{H}=sP_1+\nobreak T_{h,i,j}$ for $0\leq h\leq i\leq1\leq j\leq 3$ and $s\geq 0$ such that $h+\nobreak i+\nobreak j+\nobreak s\geq 3$
\item $\overline{H}=sP_1+\nobreak T_{0,0,0}$ for $s\geq 2$
\item $\overline{H}=sP_1+\nobreak P_t$ for $s\geq 0$ and $3\leq t\leq 7$ such that $s+\nobreak t\geq 6$
\item $\overline{H}=sP_1+\nobreak P_t+\nobreak P_u$ for $s\geq 0$, $2\leq t\leq 3$ and $3 \leq u \leq 4$ such that $s+\nobreak t+\nobreak u \geq 6$
\item $\overline{H}=sP_1+\nobreak 2P_2$ for $s\geq 1$.
\end{itemize}
\end{oproblem}
Open Problem~\ref{o-p5} shows the following.
\begin{itemize}
\item The open cases for {\sc Colouring} restricted to $(2P_2,H)$-free graphs and $(P_5,H)$-free graphs coincide.\\[-10pt]
\item The graph $H$ in each of the open cases is connected.\\[-10pt]
\item The number of minimal open cases is~$10$, namely when $\overline{H}\in \{C_3+\nobreak 2P_1,C_3+\nobreak P_2,P_1+\nobreak 2P_2\}$ (see also Section~\ref{s-intro}) and when  $\overline{H}\in \{3P_1+\nobreak P_3, 2P_1+\nobreak P_4, 2P_3, P_6, T_{0,1,1}+\nobreak P_1, T_{0,1,2}, T_{1,1,1}\}$.
\end{itemize}
As every graph~$H$ listed in Open Problem~\ref{o-p5} appears as an induced subgraph in both the graph~$G_1'$ and the graph~$G_2'$ defined in Section~\ref{s-hard}, we need new arguments to solve the open cases in Problem~\ref{o-p5}.

The complexity of {\sc Colouring} for $(K_{1,3},H)$-free graphs is less clear.
As mentioned in Section~\ref{s-intro}, the cases where $H\in \{4P_1,2P_1+\nobreak P_2,\overline{C_4+P_1}\}$ are still open.
Moreover, $K_{1,3}$ is not almost classified, as the case $H=P_t$ is open for all $t\geq 6$ (polynomial-time solvability for $t=5$ was shown in~\cite{Ma13}).
Note that $|E(\overline{H})|$ may be arbitrarily large, while Open Problem~\ref{o-p5} shows that $|E(\overline{H})|\leq 8$ in all open cases for the~$P_5$.
Since we have no new results for the case $H_1=K_{1,3}$, we refer to~\cite{LP14} for further details or to the summary of {\sc Colouring} restricted to $(H_1,H_2)$-free graphs in~\cite{GJPS}.

\end{document}